\newif\ifarxiv
\arxivtrue          

\ifarxiv
  \documentclass[aps,prd,reprint,amsmath,amssymb,superscriptaddress,nofootinbib,floatfix]{revtex4-2}
\else
  \documentclass[pdflatex,sn-mathphys-num]{sn-jnl}
\fi

\usepackage{amsmath,amssymb,amsfonts,mathtools}
\ifarxiv
  \usepackage{amsthm}
\fi
\usepackage{bm}
\usepackage{booktabs}
\usepackage{microtype}

\ifarxiv
  \providecommand{\botrule}{\bottomrule}
\fi

\DeclareMathOperator{\Tr}{Tr}

\ifarxiv
  \theoremstyle{plain}
  \newtheorem{theorem}{Theorem}
  \newtheorem{proposition}{Proposition}
  \newtheorem{corollary}{Corollary}
  \newtheorem{axiom}{Axiom}
  \newtheorem{principle}{Principle}

  \theoremstyle{remark}
  \newtheorem{remark}{Remark}

  \theoremstyle{definition}
  \newtheorem{definition}{Definition}
\else
  \theoremstyle{thmstyleone}
  \newtheorem{theorem}{Theorem}

  \newtheorem{axiom}{Axiom}

  \theoremstyle{thmstyletwo}
  \newtheorem{remark}{Remark}

  \theoremstyle{thmstylethree}
  \newtheorem{definition}{Definition}
\fi

\begin{document}

\ifarxiv
  \title{The Physics of Unresolved Uncertainty:\\
  Quantum Mechanics as a Theory of Potentiality}
  \author{Lionel Martellini}
  \email{lionel.martellini@edhec.edu}
  \affiliation{EDHEC Quantum Institute, EDHEC Business School, Nice, France}

  \begin{abstract}
We propose a reformulation of quantum mechanics as a theory of unresolved uncertainty. This theory of potentiality is formulated in the language of complex-valued measure theory, regarded as a pre-probabilistic counterpart of ordinary probability theory. In this formulation, additivity, conditioning, independence, mixtures, transition kernels, and temporal divisibility retain natural linear forms at the potentiality level, while non-classical probability-level features such as interference arise from the nonlinear Born map. Measurement is described as Bayesian-type conditioning of potentialities on actualized information, and non-selective measurement as the replacement of coherent potentiality by statistical mixtures of conditional potentiality branches. Mixed states, decoherence, composite systems, entanglement, and Bell-type correlations are also given a unified potentiality-level interpretation. The density matrix is interpreted as a coherence kernel whose off-diagonal blocks encode retained phase relations. For pure bipartite states, potentiality independence is shown to be equivalent to factorization of the Born distribution in every pair of local contexts. The resulting formulation is empirically equivalent to standard quantum mechanics, but it makes explicit a pre-probabilistic description of physical reality that is usually implicit in the Hilbert-space formalism.
  \end{abstract}

\else
  \title[Quantum Mechanics as a Theory of Potentiality]{The Physics of Unresolved Uncertainty: Quantum Mechanics as a Theory of Potentiality}

  \author[1]{\fnm{Lionel} \sur{Martellini}}\email{lionel.martellini@edhec.edu}

  \affil[1]{\orgdiv{EDHEC Quantum Institute},
  \orgname{EDHEC Business School},
  \orgaddress{\city{Nice}, \country{France}}}

  \abstract{
We propose a reformulation of quantum mechanics as a theory of unresolved uncertainty. This theory of potentiality is formulated in the language of complex-valued measure theory, regarded as a pre-probabilistic counterpart of ordinary probability theory. In this formulation, additivity, conditioning, independence, mixtures, transition kernels, and temporal divisibility retain natural linear forms at the potentiality level, while non-classical probability-level features such as interference arise from the nonlinear Born map. Measurement is described as Bayesian-type conditioning of potentialities on actualized information, and non-selective measurement as the replacement of coherent potentiality by statistical mixtures of conditional potentiality branches. Mixed states, decoherence, composite systems, entanglement, and Bell-type correlations are also given a unified potentiality-level interpretation. The density matrix is interpreted as a coherence kernel whose off-diagonal blocks encode retained phase relations. For pure bipartite states, potentiality independence is shown to be equivalent to factorization of the Born distribution in every pair of local contexts. The resulting formulation is empirically equivalent to standard quantum mechanics, but it makes explicit a pre-probabilistic description of physical reality that is usually implicit in the Hilbert-space formalism.
  }

\fi

\maketitle

\ifarxiv
\else
  \clearpage
\fi

\section{Introduction}

While the Hilbert-space formulation of quantum mechanics is empirically successful and mathematically elegant, its abstraction tends to obscure what the theory is meant to describe. In particular, foundational questions remain concerning the physical status of the state vector and the special role assigned to measurement in the standard axioms. Moreover, many key aspects of the quantum formalism, such as superposition and entanglement, remain counterintuitive. 

The perspective adopted in this paper is that quantum mechanics becomes conceptually clearer when it is formulated as a theory of \emph{unresolved uncertainty}. More precisely, we argue that the distinctive structure of quantum theory lies not merely in its capacity to assign probabilities to actualized measurement outcomes, but in its representation of a deeper pre-probabilistic level of physical description, namely the structure of potentialities before uncertainty has been resolved. The aim of the present paper is to make this perspective mathematically precise by reformulating the quantum formalism at the pre-probabilistic level, in terms of what we call a \textit{theory of potentiality}. This distinction is already latent in the standard formalism, where the basic object is not a probability distribution over actualized outcomes, but a complex-valued object whose squared modulus yields observable probabilities after application of the Born rule. The central proposal of this paper is to make this distinction explicit by reformulating quantum mechanics in terms of normalized complex-valued potentiality distributions. 

Mathematically, the natural language for such a reformulation is complex-valued measure theory. Classical probability represents uncertainty over alternatives already organized as mutually exclusive elements of a Boolean event structure. Quantum amplitudes, on the other hand, describe coherent alternatives whose separation into such an event structure is  not given ex-ante before a record is formed, and depends on the measurement context. Interference arises when unresolved amplitudes are combined before a context-specific Born distribution is formed. More generally, we argue that the quantum formalism becomes particularly transparent when familiar probabilistic notions are transposed to the pre-probabilistic level. Distributions, marginals, conditioning, independence, mixtures, correlations, transition kernels, and the description of composite systems all admit potentiality-level counterparts. Measurement can then be described as the conditioning of potentialities on actualized information. In the non-selective case, where a measurement has occurred but the outcome is not retained, coherent potentiality is replaced by a statistical mixture of conditional potentiality branches, giving a natural account of the transition from pure to mixed states and of the suppression of interference under decoherence.

The present work is related to a long-standing line of thought, going
back at least to Heisenberg, according to which the quantum state does
not describe actual properties of the physical world, but rather a structure of
\emph{potentialities}, \emph{tendencies}, or \emph{propensities} for
actualization. Heisenberg~\cite{Heisenberg1958} expressed this point explicitly when he wrote
that ``the probability wave [...] was a quantitative version of the old
concept of \textit{potentia} in Aristotelian philosophy'' and that it
introduced ``something standing in the middle between the idea of an
event and the actual event''. The present paper can be regarded as providing a mathematical analysis of this reading of the quantum formalism. More recent developments include Del Santo and Gisin's programme of potentiality realism~\cite{DelSantoGisin2023}, de Ronde's analysis of immanent powers~\cite{deRonde2017}, Kastner, Kauffman, and Epperson's dual ontology of \emph{res potentia} and \emph{res extensa}~\cite{Kastner2018}, and Aerts' work on potentiality states in quantum and cognitive structures~\cite{Aerts2009}. These approaches differ substantially, but they all share the idea that the quantum state encodes a deeper description of reality from which actual events eventually emerge. The contribution of the present paper is to move this potentiality reading from the interpretive level to the mathematical level. In most existing discussions, the Hilbert-space formalism is taken as given and potentiality is added as an ontological interpretation of that formalism. Here, by contrast, we begin with complex-valued potentiality measures and show that the usual structures of quantum mechanics admit a systematic reformulation at this pre-probabilistic level.

Several existing mathematical frameworks are close in spirit to this program. Quasi-probability distributions, such as the Wigner, Husimi, and Glauber--Sudarshan representations, express quantum states in terms of distributions that need not exhibit ordinary positive probabilities~\cite{Wigner1932,Husimi1940,Glauber1963,Sudarshan1963}. Feynman's path integral formulation is also closely related, since it assigns complex amplitudes to histories and sums them before applying a quadratic probability rule~\cite{Feynman1948}. On the measurement side, Bayesian readings of state update have been emphasized in QBism~\cite{Fuchs2014}, while the L\"uders rule provides the standard Hilbert-space expression of conditional state updating after measurement~\cite{Luders1951}. The present framework differs from these perspectives in that it treats complex-valued potentiality distributions as the primary mathematical objects and ordinary probabilities as their emerging Born-level actualized counterparts. Even more closely related is Sorkin's
quantum measure theory~\cite{Sorkin1994,Sorkin1997}, which organizes
generalized probability theories into a hierarchy of interference sum
rules: classical probability is the theory in which the second-order
interference functional vanishes, and quantum mechanics is the next
level, in which second-order interference is permitted but the
third-order functional \(I_3\) vanishes identically, a prediction
tested in the triple-slit experiment of Sinha et
al.~\cite{Sinha2010}. Sorkin's basic object is, however, a
\emph{real-valued} measure on sets of histories, and the failure of
additivity is taken as primitive, with the rule \(I_3=0\) being treated as an axiom characterizing the quantum measure. In the present formulation, the primitive object is instead a
\emph{complex-valued} measure that is fully Kolmogorov-additive, the
Born map is kept separate as the interface between potentiality and
probability, and the vanishing of third-order interference is then
\emph{derived} from the additivity
preserved at the potentiality level composed with the quadratic Born
map (see the Remark in Section~\ref{sec:potentiality_theory}). In the same spirit, Youssef's complex probability
theory~\cite{Youssef1991,Youssef1994,Youssef1995} also takes
complex-valued, additively composing measures as basic, motivated by
a Cox-style consistency analysis of Bayesian inference. This framework
is close to ours from the mathematical standpoint, but there is an
important difference in interpretation. Youssef reads the complex assignments as exotic
\emph{probabilities} of outcomes that are actual. For example, in the
two-slit experiment, the particle is taken to pass through one slit or
the other, with ordinary probability theory failing for these actual
trajectories. The theory of potentiality presented in this paper, on the other hand, is not a
probability theory of actual events with different axioms; it is a
pre-probabilistic description of alternatives that are \emph{not yet
actual}. In other words, the mathematical and conceptual framework we develop has no counterpart in the
complex-probability formalism, because in that formalism the distinction
between unresolved and resolved uncertainty does not exist. A further closely related framework is Belavkin's Eventum Mechanics and his quantum stochastic theory of continuous measurement and filtering~\cite{Belavkin1990Posterior,Belavkin2002,Belavkin2003Eventum,Belavkin2007Eventum}. In this approach, the measurement problem is addressed by distinguishing the algebra of quantum dynamical observables from the algebra of actual classical observables, or beables, and by imposing a nondemolition causality principle according to which recorded past events must be compatible with the potential quantum future. One key difference with our approach is that Eventum Mechanics develops this distinction in the quantum-stochastic language of continuous measurement, trajectories, and filtering, whereas the present paper formulates it in terms of complex-valued potentiality measures. One last related framework is the
``Events-Trees-Histories (ETH) approach'' to quantum mechanics developed by Fr\"ohlich and
collaborators~\cite{BlanchardFroehlichSchubnel2016,Froehlich2019,
FroehlichGangPizzo2025}. The ETH
approach is framed in the Heisenberg picture with a decreasing net of
observable algebras, and yields a non-linear, non-unitary stochastic law
for state evolution. In contrast, the present formulation works at the level of
complex-valued potentiality measures with linear, norm-preserving
transport, the non-linearity being induced by the Born map. A key difference is that the ETH approach modifies the dynamical law and is therefore not empirically equivalent to standard quantum mechanics, whereas the present reformulation is. 

A potential objection to the present proposal is that it may amount to a mere reinterpretation of the usual wave function, with the complex amplitude $\psi$ being simply renamed a potentiality distribution, while the mathematical content of the theory remains unchanged. While we acknowledge that our effort is to make explicit a level of interpretation that we believe has always been implicit in the quantum formalism, this objection ultimately touches upon an important distinction to be made between a reinterpretation and a reformulation of quantum mechanics. A reinterpretation leaves the formal structure of the theory untouched and attaches a new meaning to objects taken as primitive. A reformulation, by contrast, changes which notions are primitive and which are derived, even when the resulting empirical content is the same. In light of this distinction, the aim of the present paper is not merely to interpret the Hilbert-space state as a state of potentiality, but to reorganize the formalism around potentiality-level notions. As indicated in Table~\ref{tab:dictionary} from Section~\ref{sec:implications}, the translation between the potentiality language and the
Hilbert-space language is, in fact, asymmetric because the standard formalism has no names or rules for some of the pre-probabilistic operations, such as potentiality-level conditioning as distinct from L\"uders projection, or 
potentiality independence as strictly stronger than probabilistic independence.

The rest of the paper is organized as follows. Section~\ref{sec:potentiality_theory} introduces the mathematical theory of potentiality based on complex-valued measures and develops potentiality-level analogues of distributions, additivity, conditioning, independence, mixtures, and transition kernels. Section~\ref{sec:qm_as_potentiality} reformulates the basic axioms of quantum mechanics in potentiality language, including dynamics (in particular its intrinsic, clock-free
form), measurement, and contextual transformations. Section~\ref{sec:extensions} extends the formalism to mixed states and composite systems, and then treats entanglement and decoherence, Bell-type correlations being taken up with the discussion of realism in Section~\ref{sec:implications}. The conclusion summarizes the pedagogical and ontological significance of the formulation and discusses possible further extensions.

\section{The Theory of Potentiality as a Complex-Valued Measure Theory} \label{sec:potentiality_theory} 

This section introduces the mathematical framework used below for a \emph{theory of potentiality}. Both probability theory and potentiality theory provide ex ante descriptions of a physical system, before the outcome of an experiment is known. Even for a fixed measurement context, however, they describe different structures. A probability distribution assigns Born-rule weights to the context-specific outcomes and thereby specifies their expected frequencies once those outcomes are resolved, treating them as mutually exclusive elements of a Boolean event space. A potentiality state, by contrast, describes the coherent relations among the corresponding unresolved alternatives, including the information that determines how they may combine and interfere before an outcome is actualized. 

Mathematically, the appropriate framework is that of complex-valued measures, which extend ordinary measure theory from positive real measures to countably additive complex measures \cite{Halmos1950,Rudin1987,Yosida1980}. Our main purpose in this section is to show that the standard tools familiar from probability theory, such as distributions, conditioning, independence, or transition kernels, all admit natural counterparts at the level of potentialities. These counterparts preserve the linear and additive structure of the classical theory, while preparing the ground for the representation of quantum mechanics developed later. 

For clarity, we cast the analysis in settings where the potentiality spaces are assumed to be discrete and finite. Extensions to countably infinite spaces require additional summability assumptions, while extensions to continuous state spaces require an assumption that the relevant complex-valued potentiality measures admit suitable densities with respect to a reference measure and that the corresponding integrability and normalization conditions are satisfied. 

\subsection{Probability theory and potentiality theory} \label{subsec:prob_vs_pot}

Let $(\Omega,\mathcal A)$ be a measurable space, with $\Omega$ finite
and $\mathcal A$ a sigma-algebra of events. In ordinary probability theory, an uncertain system is described by a
probability measure
\[
\mathbb P:\mathcal A\to[0,1],
\]
satisfying positivity, sigma-additivity, and normalization $\mathbb P(\Omega)=1$.
Equivalently, in the discrete setting, one may work with the density
on atoms,
\[
p(\omega):=\mathbb P(\{\omega\}),
\qquad
\sum_{\omega\in\Omega}p(\omega)=1,
\]
so that, for any event $A\in\mathcal A$, the measure is recovered by
additive aggregation:
\[
\mathbb P(A)=\sum_{\omega\in A}p(\omega).
\]
If $X:\Omega\to\mathcal X$ is a discrete random variable with outcome
space $\mathcal X=\{x_1,\dots,x_n\}$, its distribution is the
pushforward
\begin{equation}
p_X(x_i):=\mathbb P\!\bigl(X^{-1}(\{x_i\})\bigr)
=\sum_{\omega:\,X(\omega)=x_i}p(\omega),
\qquad
\sum_{i=1}^{n}p_X(x_i)=1.
\label{eq:prob_distribution_section}
\end{equation}
The distribution $p_X$ characterizes the statistics of $X$ once the
uncertainty has been resolved, on the assumption that the values
$x_i$ exhaust the mutually exclusive ways in which this resolution
can occur.
 
Potentiality theory, on the other hand, provides a description of uncertainty before its resolution. It is built from the same three steps (measure,
density, aggregation), with the key difference that the weights are complex, and their normalization is quadratic. An uncertain system
is thus described by a complex measure
\[
\mathcal P:\mathcal A\to\mathbb C,
\]
satisfying sigma-additivity and the calibration condition
\[
\sum_{\omega\in\Omega}\bigl|\mathcal P(\{\omega\})\bigr|^2=1 ,
\]
imposed not at the atomic level (in general $\mathcal P(\Omega)\neq 1$). Equivalently,
one may work with the density on atoms,
\[
\psi(\omega):=\mathcal P(\{\omega\}),
\qquad
\sum_{\omega\in\Omega}\lvert\psi(\omega)\rvert^2=1,
\]
so that, for any event $A\in\mathcal A$, the measure is recovered by
the same additive aggregation:
\begin{equation}
\mathcal P(A)=\sum_{\omega\in A}\psi(\omega),
\qquad A\in\mathcal A,
\label{eq:complex_measure_density}
\end{equation}
which we refer to as the \emph{coherent aggregate} of the
potentialities associated with $A$. 

The two levels of description are connected by the Born rule: when the uncertainty is
resolved, the probability associated with the resolved elementary alternative $\omega$
is the squared modulus $p(\omega)=\lvert\psi(\omega)\rvert^2$. The elements of
$\Omega$ are the mutually exclusive elementary alternatives of the
model, and resolution occurs at their level; a coarse-graining by a
variable $X:\Omega\to\mathcal X$ therefore merely groups resolved
weights, so that the distribution of $X$ after resolution is the
\emph{incoherent} pushforward
\begin{equation}
p_X(x)=\sum_{\omega:\,X(\omega)=x}\lvert\psi(\omega)\rvert^2 .
\label{eq:incoherent_pushforward}
\end{equation}

The distinctive feature of potentiality theory is the coherent aggregate
(\ref{eq:complex_measure_density}), which sums potentiality densities before
taking squared moduli rather than after, leading to interference phenomena discussed in Section~\ref{subsec:additivity}. The reasons for
adopting this model of uncertainty in the first place are physical
and are deferred to Section~3 and Section~4.

\subsection{Potentiality additivity and interference}
\label{subsec:additivity}

The key feature of classical probability theory is additivity. If \(\{A_n\}_{n\ge 1}\subset\mathcal A\) are pairwise disjoint, then
\begin{equation}
\mathbb P\!\left(\bigcup_{n=1}^{\infty}A_n\right)=\sum_{n=1}^{\infty}\mathbb P(A_n).
\label{eq:prob_additivity}
\end{equation}
This is the third Kolmogorov axiom. 

At the potentiality level, the additive structure remains explicit:
\begin{equation}
\mathcal P\!\left(\bigcup_{n=1}^{\infty}A_n\right)=\sum_{n=1}^{\infty}\mathcal P(A_n)
\label{eq:potentiality_additivity}
\end{equation}
for every pairwise disjoint family \(\{A_n\}_{n\ge 1}\subset\mathcal A\); in the finite setting considered here, only finitely many of these sets can be nonempty. The linear structure is thus preserved at the level of complex potentialities. If unresolved alternatives contribute indistinguishably to one final outcome, the subsequent quadratic Born map generates the interference terms. If the corresponding alternatives are instead resolved into distinguishable records, the interference term is suppressed and their Born probabilities add. Interference is therefore not a failure of Kolmogorov additivity on the event algebra of an implemented measurement context. It reflects the fact that the unresolved routes have not become disjoint recorded events before their potentialities are combined.

\begin{remark}
The quadratic Born map generates second-order interference but no
irreducible third-order interference. Let \(A,B,C\) be pairwise
disjoint and define
\[
\mu(D):=\lvert \mathcal{P}(D)\rvert^{2}.
\]
Writing \(a=\mathcal{P}(A)\), \(b=\mathcal{P}(B)\), and
\(c=\mathcal{P}(C)\), the additivity of \(\mathcal{P}\) gives
\[
\mu(A\cup B\cup C)=\lvert a+b+c\rvert^{2},
\]
whose expansion contains only individual terms and pairwise cross
terms. Hence Sorkin's third-order interference functional
\[
\begin{aligned}
I_{3}(A,B,C):={}&
\mu(A\cup B\cup C)
-\mu(A\cup B)
-\mu(A\cup C)
-\mu(B\cup C) \\
&+\mu(A)+\mu(B)+\mu(C)
\end{aligned}
\]
vanishes identically. Thus, \(\mu\) has the algebraic structure of an
unnormalized grade-2-additive functional of the type appearing in
Sorkin's quantum measure theory~\cite{Sorkin1994}.
In the present construction, the grade-2 sum rule is not imposed
independently but follows from the additivity of \(\mathcal{P}\) and
the quadratic form of the Born map.
\end{remark}

Whether unresolved alternatives compose coherently, so that their potentialities add
before the Born map, or incoherently, so that their probabilities add, depends on whether any
physical record has retained information distinguishing them. This condition is most naturally expressed once the reduced mixed-state description
is available; we return to it in Section~\ref{sec:mixtures}, where the density matrix is read as
a coherence kernel and the passage from coherent to incoherent composition is identified
with the disappearance of the relevant off-diagonal blocks.

\subsection{Joint, marginal, and conditional distributions}
\label{subsec:joint_marginal}

Let \(X:\Omega\to\mathcal X\) and \(Y:\Omega\to\mathcal Y\) be two
variables defined on the \textit{same} underlying potentiality space
\((\Omega,\mathcal A,\mathcal P)\). We shall call such variables
\emph{compatible}. Throughout this subsection we assume in addition that
the pair \((X,Y)\) is \emph{jointly separating}, that is, injective on
\(\Omega\), so that the joint distribution
\eqref{eq:joint_potentiality} below is a relabeling of the underlying
state. The joint potentiality
distribution is defined by
\begin{equation}
\psi(x,y)
:=
\sum_{\omega:\,X(\omega)=x,\;Y(\omega)=y}\psi(\omega).
\label{eq:joint_potentiality}
\end{equation}
The corresponding marginal potentialities are
\begin{equation}
\psi_X(x)=\sum_y \psi(x,y),
\qquad
\psi_Y(y)=\sum_x \psi(x,y).
\label{eq:marginals_potentiality}
\end{equation}
Whenever \(\psi_X(x)\neq 0\), one may define the algebraic conditional
potentiality of \(Y\) given \(X=x\) by
\begin{equation}
\psi(y\mid x)
=
\frac{\psi(x,y)}{\psi_X(x)}.
\label{eq:conditional_potentiality}
\end{equation}
It then follows immediately that
\begin{equation}
\psi_Y(y)
=
\sum_x \psi(y\mid x)\,\psi_X(x).
\label{eq:law_total_potentiality}
\end{equation}
This is the potentiality-level analogue of the classical law of total
probability. The formal structure is the same, but the objects are now
complex-valued and retain phase information that is lost after applying
the Born map.

Note that the marginal amplitudes in
\eqref{eq:marginals_potentiality} are coherent aggregates so
\[
\lvert\psi_X(x)\rvert^2
=
\left\lvert\sum_y\psi(x,y)\right\rvert^2
\]
is not, in general, the probability of the outcome \(X=x\). Under the joint-separation assumption, when \(X\) and \(Y\) are resolved simultaneously within a single context, the probability of \(X=x\), irrespective of the value of \(Y\), is obtained by incoherent aggregation over \(y\): \begin{equation} p_X(x) = \sum_y \lvert \psi(x,y) \rvert^2. \label{eq:incoherent_marginal} \end{equation}

These considerations lead to a natural amplitude-level notion of
independence, which we formulate directly in terms of factorization.

\begin{definition}[Potentiality independence]
Compatible variables \(X\) and \(Y\) are \emph{potentiality independent}
if the joint potentiality distribution factorizes,
\begin{equation}
\psi(x,y)=\alpha(x)\,\beta(y)
\label{eq:indep_pot_factor}
\end{equation}
for some Born-normalized \(\alpha\) and \(\beta\); equivalently, if the
matrix \([\psi(x,y)]\) has rank one.
\end{definition}

This notion captures the fact that resolving one variable teaches nothing about
the potentiality of the other, which is the amplitude-level counterpart of the
familiar probabilistic intuition. Potentiality independence implies
probabilistic independence after the Born map,
\begin{equation}
\psi(x,y)=\alpha(x)\beta(y)
\Longrightarrow
|\psi(x,y)|^2
=
|\alpha(x)|^2\,|\beta(y)|^2,
\label{eq:pot_implies_prob_indep}
\end{equation}
where, by \eqref{eq:incoherent_marginal}, \(|\alpha(x)|^2\) and
\(|\beta(y)|^2\) are precisely the Born marginals of \(X\) and \(Y\). On
the other hand, the converse need not hold, and two variables may be
independent at the level of Born probabilities while
remaining dependent at the level of potentialities, through relative
phases or coherent amplitude relations not captured in \(|\psi|^2\) (see Theorem~\ref{new:thm:invariance} below for more details).

\subsection{Conditioning and Bayesian updating}
\label{subsec:conditioning}

The potentiality formulation distinguishes two notions of conditioning, a distinction that has no counterpart in ordinary probability theory: algebraic conditioning and
physical measurement conditioning. In classical probability theory,
observing the event \(X=x\) amounts to restricting the prior probability
state to the subset \(\{\omega:X(\omega)=x\}\) and renormalizing. If the
prior admits a density \(p:\Omega\to[0,1]\), the posterior density is
\begin{equation}
p_{|X=x}(\omega)
=
\frac{\mathbf 1_{\{X(\omega)=x\}}p(\omega)}
{p_X(x)},
\qquad
p_X(x)=\sum_{\omega:\,X(\omega)=x}p(\omega).
\label{eq:posterior_probability_density}
\end{equation}
Conditional distributions of other variables are then obtained by
marginalization.

At the potentiality level, physical measurement conditioning has the
same structural form, but the normalization is dictated by the Born
rule. Let \(\psi:\Omega\to\mathbb C\) be a Born-normalized potentiality
distribution. If the event \(X=x\) is actualized and recorded, the
posterior potentiality state is
\begin{equation}
\psi_{|X=x}(\omega)
=
\frac{\mathbf 1_{\{X(\omega)=x\}}\psi(\omega)}
{\sqrt{p_X(x)}},
\qquad
p_X(x)
=
\sum_{\omega:\,X(\omega)=x}
|\psi(\omega)|^2.
\label{eq:posterior_potentiality_density_born_normalized}
\end{equation}
This normalization ensures that
\begin{equation}
\sum_{\omega\in\Omega}
|\psi_{|X=x}(\omega)|^2
=
1.
\label{eq:posterior_potentiality_normalized}
\end{equation}
Thus the denominator in the physical update is not the complex marginal
amplitude \(\psi_X(x)\), but the square root of the Born probability of
the actualized event.

This physical conditioning rule should be distinguished from a purely
algebraic amplitude factorization. If \(X\) and \(Y\) are compatible
variables with joint potentiality distribution \(\psi(x,y)\), then,
whenever \(\psi_X(x)\neq0\), one may define
\begin{equation}
\psi_{\rm alg}(y\mid x)
=
\frac{\psi(x,y)}{\psi_X(x)}.
\label{eq:algebraic_conditional_potentiality}
\end{equation}
This gives the exact decomposition
\begin{equation}
\psi(x,y)
=
\psi_{\rm alg}(y\mid x)\psi_X(x),
\label{eq:algebraic_factorization_potentiality}
\end{equation}
and, similarly, an algebraic Bayes formula
\begin{equation}
\psi_{\rm alg}(y\mid x)
=
\frac{
\psi_{\rm alg}(x\mid y)\psi_Y(y)
}
{\psi_X(x)}
\label{eq:algebraic_bayes_potentiality}
\end{equation}
whenever the relevant marginal amplitudes are nonzero.

Physical conditioning describes the update associated with an
actualized and recorded event. By contrast, algebraic conditioning, which preserves an exact multiplicative decomposition of joint potentialities, should be understood not as a physical measurement update, but as an amplitude-level decomposition of unresolved alternatives. It is the appropriate operation when an intermediate variable has not been actualized, since its branches must remain available for coherent recombination rather than being converted into a statistical mixture. 

In the standard Hilbert-space formalism, physical conditioning on the
recorded outcome of an ideal projective measurement is described by the
L\"uders rule~\cite{Luders1951}. If the measurement is performed but no
conditioning on its outcome is retained, the corresponding nonselective
update produces a mixture over the recorded alternatives. An
intermediate variable that has not been measured or otherwise recorded
is treated differently: no state update is applied, and the joint
amplitude is retained and propagated coherently.

The potentiality formulation makes this latter case explicit through
algebraic conditioning, which provides an amplitude-level decomposition
of unresolved alternatives while remaining conceptually distinct from a
physical measurement update. In the Hilbert-space formalism, the same
situation is handled only implicitly by leaving the joint state intact,
rather than through a separate conditioning rule. Making this operation
explicit helps remove an ambiguity that can otherwise remain hidden
between conditioning on a recorded outcome and decomposing an
unrecorded intermediate variable.

\subsection{Transition kernels}
\label{subsec:kernels}

A final tool from probability theory that naturally extends to the theory of potentiality is the notion of transition kernel. In the classical case, a transition kernel from \(\Omega_A\) to \(\Omega_B\) is a family of conditional probabilities
\begin{equation}
K(\omega_B\mid \omega_A)\ge 0,
\qquad
\sum_{\omega_B\in\Omega_B}K(\omega_B\mid \omega_A)=1,
\label{eq:classical_kernel}
\end{equation}
which transports a probability distribution \(p_A\) on \(\Omega_A\) into a distribution \(p_B\) on \(\Omega_B\) via
\begin{equation}
p_B(\omega_B)=\sum_{\omega_A\in\Omega_A}K(\omega_B\mid \omega_A)\,p_A(\omega_A).
\label{eq:classical_kernel_transport}
\end{equation}
This describes either ordinary conditional evolution on a common space or transport between two distinct spaces of description.

The potentiality-theoretic analogue is a \emph{complex transition kernel}
\begin{equation}
L(\omega_B\mid \omega_A)\in\mathbb C,
\label{eq:potentiality_kernel}
\end{equation}
which transports a potentiality distribution \(\psi_A\) on \(\Omega_A\) into a potentiality distribution \(\psi_B\) on \(\Omega_B\) by the linear rule
\begin{equation}
\psi_B(\omega_B)
=
\sum_{\omega_A\in\Omega_A}
L(\omega_B\mid \omega_A)\,\psi_A(\omega_A).
\label{eq:potentiality_kernel_transport}
\end{equation}
A complex kernel is called \emph{isometric} when
\begin{equation}
\sum_{\omega_B\in\Omega_B}
L(\omega_B\mid \omega_A)\,
\overline{L(\omega_B\mid \omega_A')}
=
\delta_{\omega_A,\omega_A'}.
\label{eq:kernel_isometry}
\end{equation}
Isometric kernels preserve the Born norm: if \(\psi_A\) is
Born-normalized, so is the transported state \(\psi_B\). Isometry is
thus the potentiality-level counterpart of the stochasticity condition
in \eqref{eq:classical_kernel}.

As in the classical case, this notion has two distinct interpretations. First, if two compatible variables \(X\) and \(Y\) are defined on the same underlying potentiality space, then the conditional amplitudes
\begin{equation}
L(y\mid x)=\psi(y\mid x)
\label{eq:kernel_same_space}
\end{equation}
define a transition kernel through the law of total potentiality
\[
\psi_Y(y)=\sum_x L(y\mid x)\,\psi_X(x).
\]

Second, and more generally, a complex kernel may relate two distinct spaces of potentiality description, even when they are not derived from variables defined on a single common underlying space. This more general interpretation is important in the quantum setting, where incompatible observables are naturally associated with different contextual spaces and are related by norm-preserving complex kernels rather than ordinary conditional probabilities. 

\subsection{Temporal divisibility and Born-induced indivisibility}
\label{subsec:temporal_divisibility}

The distinction between probability and potentiality also clarifies the
temporal structure of quantum dynamics. In classical stochastic
processes, Markovianity is expressed through the Chapman-Kolmogorov
composition law. If \(P(t,s)\) denotes the transition probability matrix
from time \(s\) to time \(t\), then for \(t_0<t_1<t_2\),
\begin{equation}
P(t_2,t_0)
=
P(t_2,t_1)P(t_1,t_0).
\label{eq:classical_markov_divisibility}
\end{equation}
Equivalently, in components,
\begin{equation}
P_{ki}(t_2,t_0)
=
\sum_j
P_{kj}(t_2,t_1)P_{ji}(t_1,t_0).
\label{eq:classical_markov_divisibility_components}
\end{equation}

At the level of potentialities, the corresponding transition object is
not a positive stochastic matrix but a complex transition kernel. For a
closed quantum system, let \(L_{ji}(t,s)\) denote the complex transition
potentiality from alternative \(i\) at time \(s\) to alternative \(j\)
at time \(t\). Since the unitary propagator satisfies
\begin{equation}
U(t_2,t_0)
=
U(t_2,t_1)U(t_1,t_0),
\label{eq:unitary_composition_law}
\end{equation}
the transition potentiality kernel is divisible:
\begin{equation}
L(t_2,t_0)
=
L(t_2,t_1)L(t_1,t_0).
\label{eq:potentiality_divisibility}
\end{equation}
Equivalently, in components,
\begin{equation}
L_{ki}(t_2,t_0)
=
\sum_j
L_{kj}(t_2,t_1)L_{ji}(t_1,t_0).
\label{eq:potentiality_divisibility_components}
\end{equation}

The associated Born transition probabilities are defined by
\begin{equation}
P_{ji}(t,s)
=
|L_{ji}(t,s)|^2.
\label{eq:born_transition_probabilities}
\end{equation}
In general, these probabilities do not satisfy the classical
Chapman-Kolmogorov law. Indeed, using
\eqref{eq:potentiality_divisibility_components}, one obtains
\begin{equation}
P_{ki}(t_2,t_0)
=
\left|
\sum_j
L_{kj}(t_2,t_1)L_{ji}(t_1,t_0)
\right|^2,
\label{eq:born_transition_full}
\end{equation}
whereas classical probabilistic divisibility would require
\begin{equation}
P_{ki}(t_2,t_0)
=
\sum_j
P_{kj}(t_2,t_1)P_{ji}(t_1,t_0).
\label{eq:classical_transition_required}
\end{equation}
The difference between these two expressions is precisely the
interference between unresolved intermediate alternatives. Expanding
the squared modulus gives
\begin{equation}
P_{ki}(t_2,t_0)
=
\sum_j
P_{kj}(t_2,t_1)P_{ji}(t_1,t_0)
+
I_{ki}(t_2,t_1,t_0),
\label{eq:born_induced_indivisibility}
\end{equation}
where
\begin{equation}
\begin{aligned}
I_{ki}(t_2,t_1,t_0)
={}&
2\operatorname{Re}\sum_{j<j'}
L_{kj}(t_2,t_1)L_{ji}(t_1,t_0)
\\
&\times
\overline{
L_{kj'}(t_2,t_1)L_{j'i}(t_1,t_0)
}.
\end{aligned}
\label{eq:temporal_interference_term}
\end{equation}

Thus, the failure of probabilistic divisibility is not a failure of
composition at the potentiality level, but a byproduct of the interference between
unresolved intermediate potentiality paths. The quantum process may
therefore appear indivisible at the probability level, even though it
remains divisible at the potentiality level. In other words, the failure of 
classical divisibility reflects the fact that the
intermediate alternatives at time \(t_1\) have not been actualized or
recorded, so their potentialities must be summed before the Born
rule is applied.

This perspective is related to, but distinct from approaches in which
indivisibility is taken as primitive in the stochastic description of
quantum systems \cite{Barandes2023,Barandes2025}. In the present
framework, standard closed-system dynamics remains divisible at the
complex potentiality level. Indivisibility appears only at the emerging probabilistic level of description after applying
the Born map. 

\section{Quantum Mechanics as a Physical Theory of Potentiality}
\label{sec:qm_as_potentiality}

The previous section introduced a mathematical theory of potentiality built from complex-valued measures and showed that many of the central structural tools of the classical theory such as additivity, marginals, conditioning, independence, Bayesian updating, and transition kernels, all survive in this more general setting, but now at the level of complex amplitudes rather than positive probabilities. In this sense, potentiality theory is a strict extension of the probabilistic framework. It provides a mathematically coherent language for describing a pre-probabilistic level of description of the physical world, from which ordinary probabilities emerge after the application of a Born-type rule. 

By itself, however, this framework is not yet a physical theory. In order to obtain a physical theory, one must introduce specific axioms specifying how potentiality states are to be interpreted physically, how observables are represented, how the state evolves, and how measurement outcomes are extracted. The purpose of the present section is to express the standard axioms in the language of potentiality theory. This allows us to show that the quantum formalism need not appear as a mysterious departure from classical probabilistic logic, but as a theory that precisely maintains this basic logic and linear structure but at the more fundamental pre-resolved potentiality level, from which probabilities emerge as a post-resolved description.

For reference, recall the standard axioms of quantum mechanics in their familiar Hilbert-space form. A state is represented by a unit vector (or a ray) in a Hilbert space, or more generally by a density matrix. Observables are represented by self-adjoint operators. The time evolution of a closed system is unitary and generated by the Schr\"odinger equation. Finally, measurement probabilities are given by the Born rule, together with the usual post-measurement state update. The question addressed here is how these same ingredients can be reformulated in the language of potentialities.

\subsection{Axioms in potentiality language}

We begin with the pure-state case in a fixed context. Physically, the elements $\omega\in\Omega$ are the potential outcomes through which the system may be actualized, prior to any record. To each such outcome the theory assigns a complex potentiality $\psi(\omega)$ rather than a mere probability.

\begin{axiom}[State]
Let \((\Omega,\mathcal A)\) be a finite potentiality space. A pure state of the system is represented by a complex-valued potentiality measure \(\mathcal P\) admitting a density \(\psi:\Omega\to\mathbb C\) normalized so that
\begin{equation}
\sum_{\omega\in\Omega} |\psi(\omega)|^2 = 1.
\label{eq:axiom_state_potentiality}
\end{equation}
Two normalized densities related by a common phase,
\begin{equation}
\psi'(\omega)=e^{i\theta}\psi(\omega),
\end{equation}
represent the same physical pure state. The relative phases among the \(\psi(\omega)\), by
contrast, are physically significant: they leave the Born probabilities of the present
context unchanged, but affect the probabilities obtained after coherent transport or
recombination, or in measurements not diagonal in the present context. 
\end{axiom}

Thus the primitive object of the theory is not a probability distribution over already
actualized outcomes, but a normalized potentiality state, given as the phase-equivalence class of
\(\psi\), in the sense of Eq.~\eqref{eq:complex_measure_density}, defined prior to
actualization.

\begin{axiom}[Physical variables]
Within a fixed context, a \emph{fine-grained} physical variable is an injective map
\begin{equation}
X:\Omega\to\mathcal X=\{x_1,\dots,x_n\},
\label{eq:axiom_variable_potentiality}
\end{equation}
whose potentiality distribution is the relabeled state,
\(\psi_X(x)=\psi(\omega)\) whenever \(X^{-1}(\{x\})=\{\omega\}\) and
\(\psi_X(x)=0\) when no alternative is mapped to \(x\), with observable
probabilities
given by the Born rule
\begin{equation}
p_X(x)=|\psi_X(x)|^2.
\label{eq:born_rule_axiom}
\end{equation}
A \emph{coarse-grained} variable is a map \(X:\Omega\to\mathcal X\) that
groups several alternatives under a single value. Its potentiality-level description is the family of
relative potentiality states
\(\{\mathbf 1_{X^{-1}(\{x\})}\,\psi\}_{x\in\mathcal X}\), and, when \(X\)
is actualized, its observable probabilities are
\begin{equation}
p_X(x)=\sum_{\omega:\,X(\omega)=x}|\psi(\omega)|^2.
\label{eq:coarse_grained_born_axiom}
\end{equation}
\end{axiom}

In this way, a fine-grained variable is described in the same formal
manner as a random variable in probability theory, except that the
induced distribution is complex-valued and therefore richer than the
final probability law. Coherent summation enters the physical axioms through transport: through coherent kernel transport across contexts in
the law of total potentiality \eqref{eq:law_total_potentiality}, and
through the dynamics of Axiom~3, both of which preserve the Born norm by
isometry. 

\begin{axiom}[Dynamics]
For a closed, unobserved system, potentiality evolves according to a
Born-norm-preserving linear dynamics,
\begin{equation}
i\hbar\,\frac{d\psi}{dt}=H\psi,
\label{eq:axiom_dynamics_potentiality}
\end{equation}
where \(H=H^\dagger\).
\end{axiom}

The role of Axiom 3 is to express the deterministic transport of
unresolved potentiality in a closed system. The Schr\"odinger equation
is not interpreted here as the evolution of an ordinary probability
distribution, but as the norm-preserving evolution of a complex
potentiality state. The condition \(H=H^\dagger\) guarantees conservation
of the Born norm,
\begin{equation}
\frac{d}{dt}\sum_{\omega\in\Omega}|\psi(\omega,t)|^2=0.
\end{equation}
Thus closed-system dynamics preserves coherent potentiality without
actualizing any particular outcome. Observable probabilities arise only
when a context is selected and the Born rule is applied. The distinction
between dynamical transport and actualization is thus that dynamics transports unresolved potentiality, whereas measurement
conditions potentiality on an actualized record, as we now discuss.

\begin{axiom}[Complete actualization and selective measurement]
If a measurement of the variable \(X\) produces a fully resolved and
retained outcome \(x\in\mathcal X\), then the prior potentiality state
\(\psi\) is updated by conditioning on the event \(\{X=x\}\). The
post-measurement state is
\begin{equation}
\psi_{\,|X=x}(\omega)
=
\frac{\mathbf 1_{\{X(\omega)=x\}}\,\psi(\omega)}
{\sqrt{\sum_{\omega':\,X(\omega')=x}
\lvert\psi(\omega')\rvert^2}}.
\label{eq:selective_measurement_update}
\end{equation}
\end{axiom}

This axiom explains how when an actualized outcome is obtained through a measurement, the potentiality state is updated conditionally according to \eqref{eq:selective_measurement_update}. This update is of Bayesian form: one restricts the prior state to the event compatible with the observed outcome and renormalizes. The difference from ordinary probability theory is that the conditioning takes place at the level of complex potentialities rather than positive probabilities. 

In this formulation, unitary evolution and state reduction are not two heterogeneous mechanisms. They are two distinct modes of transformation of the same basic object, namely the potentiality state. This picture also clarifies the status of measurement. Measurement is not treated as the mysterious production of probabilities out of an already probabilistic object. Rather, we argue that it is the transition from a coherent potentiality structure to either a conditional potentiality branch (selective measurement) or a classical mixture of such branches (non-selective measurement, discussed in Section \ref{sec:mixtures}). The role of the Born rule is then to convert this potentiality-level description into empirical probabilities.

\begin{remark}
Axiom~4 describes the ideal limiting case in which the alternatives
associated with distinct outcomes have been completely resolved into
distinguishable records. In this limit, conditioning on the recorded
event is exact and gives the potentiality-theoretic counterpart of the
L\"uders update. Actualization will be used more generally, however, for
the physical formation of a record that distinguishes alternatives to
some degree. When the record states are only partially distinguishable,
actualization is correspondingly partial: some coherence is suppressed,
but residual coherence remains available for interference. This graded
notion is developed in Section~\ref{subsubsec:emergence_classicality}.
\end{remark}

\subsection{Selective and non-selective measurement}
\label{sec:nonselective}
 
The distinction between selective and non-selective measurement becomes especially transparent in this language.
 
A \emph{selective} measurement is one in which the realized measurement outcome remains available in an accessible physical record, including through an apparatus pointer, an ancilla, an environment fragment, or a memory device. Implicit in this analysis, as in the standard formalism, is a record-consistency principle: if a physical process has left an objective record in the world, then subsequent dynamics and state assignment must be consistent with the existence of that record. This requirement is similar to Belavkin's nondemolition, or quantum-causality, condition~\cite{Belavkin1994}. 

If the event \(X=x\) is observed and recorded in this broad sense, the relevant post-measurement state is the conditional potentiality state \(\psi_{\,|X=x}\). If instead a measurement of \(X\) has taken place, but the realized value is not retained, we then talk about a \emph{non-selective} measurement. In that case, the post-measurement description is not a single conditional potentiality state but a statistical mixture of branches:
\begin{equation}
\bigl\{(\mathbb P(X=x),\,\psi_{\,|X=x})\bigr\}_{x\in\mathcal X}.
\label{eq:non_selective_branch_ensemble}
\end{equation}
The point is that the branch has been physically singled out, even if the observer does not keep track of which branch it was. Interference between distinct alternatives \(x\) therefore disappears, and the state is no longer described by one coherent potentiality distribution but by an ensemble of conditional potentiality states. This is precisely the signature of a mixed state, in which the observable probability for \(Y\) is then
\begin{equation}
\mathbb P_{\mathrm{mix}}(Y=y)
=
\sum_x \mathbb P(X=x)\,\mathbb P_x(Y=y).
\label{eq:mixed_born_potentiality}
\end{equation}

Here the mixed state appears specifically as the product of non-selective measurement; its general definition, independent of how it arises, together with the potentiality-matrix representation \(\rho=\Psi\Psi^\dagger\), is developed in Section~\ref{sec:mixtures}.

\subsection{Contextual transformations}
\label{subsec:contextual_transformations}

The preceding material introduced the basic physical expression of the potentiality formalism for a pure state in a fixed context. One of the central lessons of the quantum formalism, however, is that the same physical state admits different contextual representations, associated with different observables or different experimental arrangements, and that these representations cannot in general be embedded into one single classical event structure. In the standard Hilbert-space language, this is expressed by the existence of different bases or, more generally, different spectral decompositions related by unitary transformations. In the present framework, the same idea is naturally formulated in terms of contextual spaces and complex transition kernels acting on potentiality distributions.

A physical state is thus not represented by one probability distribution over a single universal sample space of already actualized outcomes. Rather, it is represented by a set of contextual potentiality distributions, each attached to a specific space of possible outcomes associated with a given variable or measurement arrangement. A context should therefore be understood as the space in which a given variable is resolved into its possible outcomes, prior to actualization. Let \(A\) be such a context, with outcome space \(\Omega_A\). A state represented in that context is described by a normalized complex distribution
\begin{equation}
\psi^{(A)}:\Omega_A\to\mathbb C,
\qquad
\sum_{\omega_A\in\Omega_A} |\psi^{(A)}(\omega_A)|^2=1.
\label{eq:state_in_context_A}
\end{equation}
Likewise, another context \(B\), with outcome space \(\Omega_B\), carries another representation
\begin{equation}
\psi^{(B)}:\Omega_B\to\mathbb C,
\qquad
\sum_{\omega_B\in\Omega_B} |\psi^{(B)}(\omega_B)|^2=1.
\label{eq:state_in_context_B}
\end{equation}

The key point is that \(\psi^{(A)}\) and \(\psi^{(B)}\) need not be interpreted as two unrelated states but may be two contextual representations of the same physical state. The problem is then to describe how one passes from one contextual representation to another. The potentiality-theoretic analogue of a change of basis is a complex transition kernel between two contexts.

\begin{definition}[Contextual transformation]
Let \(\Omega_A\) and \(\Omega_B\) be two contextual spaces. A \emph{contextual transformation} from \(A\) to \(B\) is a complex kernel
\begin{equation}
L_{B\leftarrow A}(\omega_B\mid\omega_A)\in\mathbb C
\label{eq:contextual_kernel_definition}
\end{equation}
acting linearly on potentiality distributions according to
\begin{equation}
\psi^{(B)}(\omega_B)
=
\sum_{\omega_A\in\Omega_A}
L_{B\leftarrow A}(\omega_B\mid\omega_A)\,\psi^{(A)}(\omega_A).
\label{eq:contextual_kernel_action}
\end{equation}
\end{definition}

This is the direct analogue of the transformation
\[
|\psi\rangle \mapsto U|\psi\rangle
\]
in the Hilbert-space formalism, written instead at the level of contextual potentiality amplitudes. In general, the kernel \(L_{B\leftarrow A}\) is required to preserve normalization, that is, to be isometric in the sense of \eqref{eq:kernel_isometry}. Under this condition, if \(\psi^{(A)}\) is normalized, then \(\psi^{(B)}\) is normalized as well, so total potentiality norm is preserved exactly as unitary changes of representation preserve the norm of a state vector in the standard formalism:
\begin{equation}
\sum_{\omega_B} |\psi^{(B)}(\omega_B)|^2
=
\sum_{\omega_A} |\psi^{(A)}(\omega_A)|^2.
\label{eq:norm_preservation_context}
\end{equation}

The potentiality formalism also makes it possible to formulate clearly the
difference between compatible and incompatible variables. Two variables
\(X\) and \(Y\) are \emph{compatible} when there \emph{exists} a single
context \((\Omega,\mathcal A,\mathcal P)\) on which both can be defined as
variables, possibly coarse-grained ones. In that case, one may form a joint
potentiality distribution \(\psi(x,y)\), define marginal and conditional
potentiality distributions, and apply the law of total potentiality exactly
as in the previous section. This existence condition is precisely
equivalent to the commutation of the corresponding self-adjoint operators: by the spectral
theorem, two such operators commute if and only if they admit a common
eigenbasis, that is, if and only if they are functions on one and the same
context. In contrast, incompatibility describes a situation where no context
supports both descriptions at once. One may still pass from one context to
another through a transformation kernel of the form
\eqref{eq:contextual_kernel_action}, but one should not interpret the two
contexts as subsets of one larger classical sigma-algebra of resolved
outcomes. 

\section{Extension of the potentiality formalism to mixed states and composite systems}
\label{sec:extensions}

We now turn to two further extensions that are necessary for a fuller physical theory, namely the inclusion of mixed states and the description of composite systems.

\subsection{Mixed potentiality states and coherent composition}
\label{sec:mixtures}

Non-selective measurements and classical uncertainty over preparation
procedures can both give rise to mixed states. We first introduce them
independently of any composite-system representation. Their interpretation as
reduced descriptions of pure states on enlarged potentiality spaces will be
developed in Section~\ref{subsubsec:reduced_states}.

\subsubsection{Mixed states and potentiality matrices}

A pure state is represented, in a fixed context, by a single normalized
potentiality vector $\psi\in\mathbb C^n$. A mixed state instead represents
either classical uncertainty over several such states or the loss of access to
a branch label created by measurement or entanglement. It may be specified by
an ensemble
\begin{equation}
\bigl\{(q_r,\psi^{(r)})\bigr\}_{r\in\mathcal R},
\qquad
q_r\geq0,
\qquad
\sum_{r\in\mathcal R}q_r=1,
\label{eq:potentiality_statistical_ensemble}
\end{equation}
where each $\psi^{(r)}$ is Born-normalized. Its operational state is
\begin{equation}
\rho
=
\sum_{r\in\mathcal R}q_r
|\psi^{(r)}\rangle\langle\psi^{(r)}|.
\label{eq:mixed_potentiality_state}
\end{equation}
The branches therefore combine statistically rather than coherently. In a
non-selective measurement, $r$ labels the possible conditional branches; in a
preparation mixture, it labels classical uncertainty about the prepared
potentiality state.

Equivalently, define the potentiality matrix
\begin{equation}
\Psi_{ir}=\sqrt{q_r}\,\psi_i^{(r)},
\qquad
\rho=\Psi\Psi^\dagger.
\label{eq:rho_from_potentiality_matrix}
\end{equation}
The density matrix is the operational object obtained when the branch label is
not retained, whereas $\Psi$ retains a branch-resolved description relative to
a specified ensemble, record structure, or purification. The map
$\Psi\mapsto\Psi\Psi^\dagger$ is a matrix-valued analogue of the Born rule in that it removes branch resolution while preserving all coherence that can affect measurements on the system.

The factorization is not unique. The state is pure precisely when $\Psi$ has
rank one; otherwise different ensembles and different purifications may yield
the same operational density matrix. Thus the columns of $\Psi$ acquire a
physical branch-level meaning only when additional structure identifies a
particular preparation, record, monitoring scheme, or auxiliary subsystem.
Without such structure, only $\rho$ is operationally fixed.

In a fixed context, the density matrix may be read as a coherence kernel: its
diagonal entries are the Born weights, while its off-diagonal entries encode
relative-phase information that may affect future statistics after coherent
contextual transformations. If the off-diagonal block between two families of
alternatives vanishes, the reduced state is invariant under relative-phase
transformations between those families, and their contributions compose
incoherently.

\subsection{Extension of the potentiality formalism to composite systems and entanglement}
\label{sec:extension_composite} 

Having extended the formalism from pure states to mixed states, we now turn to composite systems. 

\subsubsection{Description of Composite Systems}

We first introduce the axiom that provides the description of a composite system.

\begin{axiom}[Composite systems]
A composite physical system is described by a factorization of the underlying potentiality space into subsystem spaces,
\begin{equation}
\Omega=\Omega_A\times\Omega_B,
\qquad
\omega=(\omega_A,\omega_B),
\label{eq:composite_axiom}
\end{equation}
where \(\omega_A\) and \(\omega_B\) describe distinguishable subsystems or distinguishable aspects of the world.
\end{axiom}

A global pure state is then represented by a normalized complex function
\begin{equation}
\psi(\omega_A,\omega_B).
\label{eq:global_potentiality_composite}
\end{equation}

\begin{definition}[Separable and entangled states]
The global potentiality state is said to be \emph{separable} if it factorizes as
\begin{equation}
\psi(\omega_A,\omega_B)=\alpha(\omega_A)\beta(\omega_B)
\label{eq:separable_potentiality}
\end{equation}
for some partial potentiality states \(\alpha\) and \(\beta\). Otherwise, the state is said to be \emph{entangled}.
\end{definition}

A separable state describes two autonomous potentiality structures, one for each subsystem. An entangled state describes a single irreducibly joint potentiality structure that cannot be decomposed into two independent components. In this sense, entanglement is not a mysterious correlation between already actual local outcomes; it is the non-factorizability of the underlying potentiality structure itself.

It is easy to see that potentiality independence is strictly stronger
than probabilistic independence of the Born distributions (see eq. \ref{eq:pot_implies_prob_indep}). In other words, as is well known, a joint state can fail to factorize at the
amplitude level, that is $\operatorname{rank}[\psi(x,y)]\geq 2$, so that
$\psi(x,y)\neq\alpha(x)\beta(y)$ for any choice of $\alpha$ and $\beta$, while
its Born distribution nevertheless factorizes in the given context,
$|\psi(x,y)|^2=p_X(x)\,p_Y(y)$ because relative phases can destroy the product structure of
the amplitudes without leaving any trace in their squared moduli. The
following theorem provides an exact characterization of the relationship between the two notions: potentiality independence
is probabilistic independence \emph{in all local contexts
simultaneously}. 

\begin{theorem}[Contextual invariance characterization]
\label{new:thm:invariance}
A \emph{local context pair} is a pair \((U,V)\) of unitary contextual
transformations acting on the \(A\) and \(B\) factors respectively, so
that the transformed joint potentiality state is
\begin{equation}
\psi'(x',y')
=
\sum_{x,y}
U(x'\mid x)\,
V(y'\mid y)\,
\psi(x,y).
\label{eq:local_context_transformation_joint}
\end{equation}
For a normalized joint pure potentiality state on finite spaces,
\(\psi\in\ell^2(\Omega_A\times\Omega_B)\), the following statements are
equivalent:
\begin{enumerate}
\item[\rm(i)] \(\psi\) is potentiality independent:
\begin{equation}
\psi(x,y)=\alpha(x)\beta(y)
\end{equation}
for some normalized potentiality distributions \(\alpha\) and \(\beta\);

\item[\rm(ii)] for every local context pair \((U,V)\), the Born
distribution of the transformed state factorizes into its incoherent
marginals:
\begin{equation}
|\psi'(x',y')|^2
=
p'_X(x')\,p'_Y(y'),
\end{equation}
where
\begin{equation}
p'_X(x')
=
\sum_{y'}|\psi'(x',y')|^2,
\qquad
p'_Y(y')
=
\sum_{x'}|\psi'(x',y')|^2;
\end{equation}

\item[\rm(iii)] the matrix \([\psi(x,y)]\) has rank one.
\end{enumerate}
\end{theorem}

\begin{proof}
The equivalence between {\rm (i)} and {\rm (iii)} is an elementary
consequence of linear algebra. A nonzero matrix has rank one if and only
if its entries factorize as
\begin{equation}
\psi(x,y)=\alpha(x)\beta(y).
\end{equation}
Since \(\psi\) is normalized, the factors may be chosen to be
Born-normalized.

To prove {\rm (i)}\(\Rightarrow\){\rm (ii)}, suppose that
\(\psi(x,y)=\alpha(x)\beta(y)\). Under a local context pair
\((U,V)\), one has
\begin{align}
\psi'(x',y')
&=
\sum_{x,y}
U(x'\mid x)V(y'\mid y)\alpha(x)\beta(y)
\\
&=
\left(
\sum_x U(x'\mid x)\alpha(x)
\right)
\left(
\sum_y V(y'\mid y)\beta(y)
\right).
\end{align}
Thus
\begin{equation}
\psi'(x',y')
=
\alpha'(x')\beta'(y'),
\end{equation}
where \(\alpha'=U\alpha\) and \(\beta'=V\beta\). Consequently,
\begin{equation}
|\psi'(x',y')|^2
=
|\alpha'(x')|^2|\beta'(y')|^2
=
p'_X(x')p'_Y(y').
\end{equation}

We prove {\rm (ii)}\(\Rightarrow\){\rm (i)} by contraposition. Suppose
that \(\psi\) is not a product state. Its Schmidt decomposition then
contains at least two nonzero coefficients:
\begin{equation}
\psi
=
\sum_k s_k\,e_k\otimes f_k,
\qquad
s_1\geq s_2>0,
\label{eq:schmidt_decomposition_invariance}
\end{equation}
where \(\{e_k\}\subset\ell^2(\Omega_A)\) and
\(\{f_k\}\subset\ell^2(\Omega_B)\) are orthonormal families.

Choose local contextual transformations \(U\) and \(V\) such that the
transformed contexts coincide with the corresponding Schmidt contexts.
In these contexts, the transformed potentiality distribution has
components
\begin{equation}
\psi'(k,\ell)
=
s_k\,\delta_{k\ell}.
\label{eq:schmidt_context_amplitudes}
\end{equation}
Its Born distribution is therefore
\begin{equation}
p'(k,\ell)
=
s_k^2\,\delta_{k\ell},
\end{equation}
while its incoherent marginals are
\begin{equation}
p'_X(k)=s_k^2,
\qquad
p'_Y(\ell)=s_\ell^2.
\end{equation}
In particular,
\begin{equation}
p'(1,2)=0,
\end{equation}
whereas
\begin{equation}
p'_X(1)p'_Y(2)
=
s_1^2s_2^2
>
0.
\end{equation}
The Born distribution therefore fails to factorize in this local
context pair. Hence {\rm (ii)} fails whenever \(\psi\) is not a product
state, which proves {\rm (ii)}\(\Rightarrow\){\rm (i)}.
\end{proof}

The mathematical content of the equivalence is standard and is related to Gisin's theorem, according to which every entangled pure state violates a Bell inequality~\cite{Gisin1991}. The contribution of the present theorem is to identify context-universal statistical independence as exactly the probability-level content of amplitude-level independence, a question that the standard formalism does not explicitly formulate because it does not introduce a specific name to the corresponding notion on the amplitude side.

\subsubsection{Mixed states as reduced descriptions of pure states}
\label{subsubsec:reduced_states}

The contextual invariance theorem characterizes when a joint potentiality state assigns autonomous potentiality states to its subsystems. When it does not, that is, for entangled states, the question arises of what the potentiality-level description of a single subsystem is. The answer cannot be the marginal amplitude \(\sum_y\psi(x,y)\) since the projection onto the first factor groups several joint alternatives under each value of \(x\), so the coherent marginal is not Born-normalized, and its square is not the outcome distribution of any measurement on \(A\).

To define the correct reduced object first define for each alternative \(y\in\Omega_B\),
\begin{equation}
\varphi_y:=\psi(\cdot,y)\in\mathbb C^{\Omega_A}
\label{eq:relative_state_subsystem}
\end{equation}
as the unnormalized relative potentiality state of \(A\) given \(y\), with Born weight
\begin{equation}
q_y
=
\sum_x|\psi(x,y)|^2.
\label{eq:relative_state_born_weight}
\end{equation}

Writing
\begin{equation}
\varphi_y
=
\sqrt{q_y}\,\psi^{(y)}
\label{eq:normalized_relative_potentiality_state}
\end{equation}
with \(\psi^{(y)}\) normalized whenever \(q_y>0\), the family \(\{(q_y,\psi^{(y)})\}_y\) is precisely a statistical ensemble of potentiality states in the sense of \eqref{eq:potentiality_statistical_ensemble}. The associated potentiality matrix is the joint amplitude itself, read as a matrix:
\begin{equation}
\Psi_A:=\bigl[\psi(x,y)\bigr]_{x\in\Omega_A,\;y\in\Omega_B},
\qquad
\rho_A=\Psi_A\Psi_A^\dagger.
\label{eq:reduced_potentiality_matrix}
\end{equation}
Equivalently, the entries of the reduced density matrix are
\begin{equation}
(\rho_A)_{xx'}
=
\sum_{y\in\Omega_B}
\psi(x,y)\,\overline{\psi(x',y)}.
\label{eq:partial_trace_components}
\end{equation}

Equation~\eqref{eq:partial_trace_components} is precisely the partial trace over \(B\), written directly in potentiality coordinates. The reduced state of a subsystem is thus a mixed state in the sense of Section~\ref{sec:mixtures}; the partial trace of the standard formalism is recovered as marginalization followed by incoherent aggregation, and requires no additional postulate.

Statistical mixtures and reduced mixed states have different physical origins. In a statistical mixture, the label \(y\) represents classical uncertainty about which potentiality state was prepared. In a reduced mixed state, by contrast, the complete joint potentiality state may be perfectly known and pure. The mixedness of \(A\) arises because the correlations between \(A\) and \(B\) have been omitted from the local description. Measurements performed on \(A\) alone cannot distinguish between these two origins when they give rise to the same density matrix. Operationally, both are represented by the same \(\rho_A\).

Conversely, it is well known that every density matrix can be represented as the reduced state of a pure potentiality state on a sufficiently large auxiliary space. Let \(\rho_A\) be any positive semidefinite matrix satisfying
\begin{equation}
\Tr(\rho_A)=1.
\end{equation}
It admits a factorization
\begin{equation}
\rho_A
=
\Psi_A\Psi_A^\dagger
\label{eq:general_density_matrix_factorization}
\end{equation}
for some matrix \(\Psi_A\). For example, if
\begin{equation}
\rho_A
=
V\Lambda V^\dagger
\label{eq:density_matrix_spectral_decomposition}
\end{equation}
is a spectral decomposition, where \(\Lambda\) is diagonal with nonnegative entries, one may take
\begin{equation}
\Psi_A
=
V\Lambda^{1/2}.
\label{eq:canonical_purification_matrix}
\end{equation}
The entries \((\Psi_A)_{xr}\) may then be interpreted as a normalized joint potentiality distribution on an enlarged space \(\Omega_A\times\Omega_R\), since
\begin{equation}
\sum_{x,r}|(\Psi_A)_{xr}|^2
=
\Tr(\Psi_A\Psi_A^\dagger)
=
\Tr(\rho_A)
=
1.
\label{eq:purification_normalization}
\end{equation}
Applying the reduction rule \eqref{eq:partial_trace_components} recovers the original density matrix. Mixed states may therefore always be regarded as reduced descriptions of pure potentiality states on an enlarged space, although neither the purification nor the corresponding branch decomposition is unique. These statements are, of course, not new; they represent the potentiality formulation of well-known results from the standard formalism (see
the Hughston--Jozsa--Wootters theorem \cite{HJW1993}). 

\subsubsection{Emergence of classicality: decoherence and incoherent aggregation}
\label{subsubsec:emergence_classicality}

The potentiality formalism also provides a natural language for the emergence of classical behavior, with decoherence being driven by the system becoming entangled with an environment whose records are not retained. Consider a system with outcome space \(\Omega_S=\{x_1,\dots,x_n\}\) coupled to an environment \(E\), so that the composite potentiality space factorizes as \(\Omega=\Omega_S\times\Omega_E\). Suppose the interaction correlates each system alternative \(x_i\) with a normalized environmental potentiality state \(\chi_i\), so that an initially coherent system state with amplitudes \(\psi_i\) evolves into the entangled joint state
\begin{equation}
\Psi(x_i,\omega_E)
=
\psi_i\,\chi_i(\omega_E),
\qquad
\sum_{i=1}^n|\psi_i|^2=1.
\label{eq:system_environment_entangled}
\end{equation}
By Section~\ref{subsubsec:reduced_states}, the reduced description of the system is the potentiality matrix \(\Psi_S=[\psi_i\,\chi_i(\omega_E)]\), whose observable content, \((\rho_S)_{ij}=\psi_i\overline{\psi_j}\,\gamma_{ij}\), depends on the environmental branches only through their overlaps,
\begin{equation}
\gamma_{ij}
:=
\sum_{\omega_E}\chi_i(\omega_E)\,\overline{\chi_j(\omega_E)},
\qquad
\gamma_{ii}=1.
\label{eq:record_overlaps}
\end{equation}

The Hermitian matrix \(\gamma\) interpolates between two extremes: \(\gamma_{ij}\equiv1\) when the environment does not resolve the system alternatives, all \(\chi_i\) being equal, and
\begin{equation}
\gamma_{ij}=\delta_{ij}
\label{eq:orthogonal_records}
\end{equation}
when the environmental records are orthogonal, hence distinguishable in principle. The probability of a measured outcome for a recombined system context \(U(u\mid x_i)\) extended to the environment is
\begin{equation}
P(u)
=
\sum_i |U(u\mid x_i)|^2\,|\psi_i|^2
+
2\operatorname{Re}
\sum_{i<j}
U(u\mid x_i)\,\overline{U(u\mid x_j)}\,
\psi_i\overline{\psi_j}\,\gamma_{ij}.
\label{eq:decoherence_interference_split}
\end{equation}

The interference term is linear in the record overlaps and one may distinguish two polar cases. On the one hand, for \(\gamma_{ij}\equiv1\) the two terms can be grouped into \(\bigl|\sum_i U(u\mid x_i)\psi_i\bigr|^2\), and the system behaves as a coherent pure state. On the other hand, for \(\gamma_{ij}=\delta_{ij}\), the interference term vanishes and
\begin{equation}
P(u)
=
\sum_i |\psi_i|^2\,|U(u\mid x_i)|^2,
\label{eq:decohered_probabilities}
\end{equation}
which is exactly the non-selective measurement formula \eqref{eq:mixed_born_potentiality} with branch weights \(\mathbb P(x_i)=|\psi_i|^2\). An environment carrying orthogonal branch records therefore induces, for the reduced system, the same operational state as a non-selective measurement. Decoherence can thus be regarded as the dynamical passage from the first regime to the second, driven by the growth of distinguishability of the environmental branches, that is, by the decay of \(\gamma_{ij}\) towards \(\delta_{ij}\).

The overlap \(|\gamma_{ij}|\) gives a precise meaning to \emph{partial
actualization}. When the correlated records are orthogonal, \(\gamma_{ij}=0\), the
off-diagonal block of the reduced description vanishes, the alternatives are fully resolved
into distinguishable facts, and the actualization is complete and, once the records are
redundant and stable, effectively irreversible. When the records are only partially
distinguishable, \(0<|\gamma_{ij}|<1\), a corresponding part of the coherence has been removed while the residual coherence
\(|\psi_i\psi_j\gamma_{ij}|\) is still able to interfere. The alternatives are then neither
fully actualized nor fully unresolved. In this sense, the actualization process is a continuum controlled by the distinguishability of the record, and it reduces to complete
actualization only in the orthogonal limit. 

This partial notion of actualization coincides with the distinction between strong and weak measurement. A
strong, projective measurement forms an orthogonal record, \(\gamma_{ij}=0\), and actualizes
completely; a weak-coupling measurement forms an imperfectly distinguishable record,
\(|\gamma_{ij}|\to1\), and actualizes only partially, extracting little information and leaving
most of the coherence intact. In this sense, strong and weak measurements are endpoints of the single graded actualization process. 

\section{Ontological and Physical Implications}
\label{sec:implications}

We now discuss two kinds of
implications of the potentiality formalism. The first concerns the possible ontological reading of the
quantum state. The second concerns possible physical and mathematical
extensions.

\subsection{Ontological and conceptual value}

The potentiality formulation clarifies the possible ontological reading
of the quantum state. In the standard formalism, the quantum state is
typically introduced as a vector or density operator whose empirical role is
to generate probabilities, thus implicitly suggesting that the state is
merely a computational tool or an encoding of knowledge. The potentiality
formulation suggests a different interpretation, where the quantum state describes an
objective structure of unresolved physical potentialities, from which
actual outcomes emerge through contextual resolution. This interpretation is naturally close to the Aristotelian distinction between
potentiality and actuality. What is physically real encompasses not only what is already actual but also, according to quantum mechanics, the structure of what may become actual, represented by complex-valued potentiality distributions whose phases determine how
unresolved alternatives recombine before the Born rule is applied. One could argue that the presence of observable interferences is the physical witness that potentialities are a fundamental aspect of the description of the physical world. If what has observable consequences is naturally regarded as real, a
criterion dating back to Plato's Eleatic Stranger for whom the defining condition of being is the power to act or be acted upon \cite{PlatoSophist}, then
potentialities qualify since relative phases among unresolved alternatives affect
observable outcome statistics in ways that no single-context probability distribution can encode. At the same time, the formalism does not force this ontological
interpretation. One may also read it instrumentally, as a clearer reformulation of the standard quantum rules. The main point of this paper is that the
notion of potentiality is not merely metaphorical. It can be given
precise mathematical content through complex-valued measures,
Born-normalized potentiality distributions, conditioning rules, contextual transformations, and the distinction
between coherent transport and incoherent aggregation. In other words, quantum mechanics does not compel us to say that potentialities are real, but it gives a precise mathematical meaning to the idea that reality may include more than what is already actual.

It should be emphasized that the potentiality framework is not a hidden-variable theory in
the classical sense. What is real prior to measurement is not a list of
pre-existing values for all observables, but a structured space of
possible actualizations. Actuality is thus not the primitive level of
the theory. It is the emerging result of a process in which a context is selected,
a record is formed, and the potentiality structure is conditioned and
partially decohered. This interpretation helps distinguish different forms of realism. Bell-type violations exclude locally causal, measurement-independent models in which the relevant outcomes are fixed by pre-existing local variables~\cite{Bell1964,CHSH1969}. They do not exclude every form of realism, including explicitly nonlocal realist theories, and they do not by themselves rule out realism at the level of potentialities. In the present language, entangled systems are not systems endowed with pre-existing local values; they are systems whose joint potentiality state does not factorize into autonomous local potentiality states. The framework may therefore support a form of potentiality realism.

On a different note, we argue that the potentiality formulation may also have some conceptual and pedagogical value.
It begins with notions familiar from ordinary probability theory, such
as distributions, conditioning, independence, mixtures, and transition
kernels, but transports them to a pre-probabilistic level. In classical
probability theory, these structures describe alternatives organized within a Boolean event space. In quantum mechanics, the corresponding structures apply instead at the
level of complex-valued potentialities before the Born rule turns them
into observable probabilities. In this context, the Hilbert space representation is naturally reinterpreted as the canonical mathematical geometry of complex
potentialities. 

It is useful to close this section by drawing the correspondence
between the standard Hilbert-space vocabulary and the potentiality vocabulary developed
in this paper. Table~\ref{tab:dictionary} does so in two blocks. The first block lists
the standard notions together with their potentiality-level counterparts. The second block lists notions that are primitive and natural at the
potentiality level but possess no named counterpart in the standard formalism. The upper block shows that every structural ingredient of the standard formalism admits an explicit representation in
the potentiality language, which is the content of the empirical equivalence claimed
throughout this paper. On the other hand, the lower block shows that the reverse
translation is not complete since the potentiality language contains named, rule-governed operations
whose standard counterparts are implicit steps of calculation or, in the case
of the independence gap, questions the formalism does not pose. The
asymmetry of this correspondence dictionary supports the claim that the present proposal is a reformulation rather than a mere relabeling that would lead to a bijective dictionary. 

\ifarxiv
\begin{table*}[t]
\else
\begin{table}[ht]
\fi
\caption{Dictionary between the standard formalism and the theory of potentiality.
The upper block is a faithful translation; in the lower block the standard column is
empty: the operations are performed in standard calculations but are not named objects
of the formalism.}
\label{tab:dictionary}
\ifarxiv
\begin{tabular}{@{}p{0.44\textwidth}p{0.48\textwidth}@{}}
\else
\begin{tabular*}{\textwidth}{@{\extracolsep{\fill}}p{0.44\textwidth}p{0.48\textwidth}@{}}
\fi
\toprule
Standard formalism & Theory of potentiality \\
\midrule
Unit vector (ray) $|\psi\rangle\in\mathcal H$ &
Potentiality state $\psi:\Omega\to\mathbb C$,
$\sum_\omega|\psi(\omega)|^2=1$ (Axiom~1) \\[3pt]
Orthonormal basis $\{|\omega\rangle\}$ &
Context: space $\Omega$ of atomic alternatives \\[3pt]
Nondegenerate observable $\hat A=\sum_x x\,|x\rangle\langle x|$ &
Fine-grained variable $X:\Omega\to\mathcal X$ injective; relabeled state $\psi_X$
(Axiom~2) \\[3pt]
Projected vector $\Pi_A|\psi\rangle$ (degenerate outcome) &
Relative potentiality state $\mathbf 1_A\,\psi$ of a non-atomic event \\[3pt]
Born rule $p(x)=|\langle x|\psi\rangle|^2$ &
$p_X(x)=|\psi_X(x)|^2$; coarse-grained case by incoherent aggregation
\eqref{eq:incoherent_pushforward} \\[3pt]
L\"uders update $\Pi_A|\psi\rangle/\|\Pi_A|\psi\rangle\|$ &
Physical conditioning
\eqref{eq:posterior_potentiality_density_born_normalized} \\[3pt]
Unitary evolution; change of basis $U$ &
Isometric kernel transport
\eqref{eq:potentiality_kernel_transport}, \eqref{eq:kernel_isometry} \\[3pt]
Composition of propagators $U(t_2,t_0)=U(t_2,t_1)U(t_1,t_0)$ &
Divisibility of transition potentialities \eqref{eq:potentiality_divisibility} \\[3pt]
Density matrix $\rho$ &
$\rho=\Psi\Psi^\dagger$ for a potentiality matrix $\Psi$
\eqref{eq:rho_from_potentiality_matrix} \\[3pt]
Partial trace $\rho_A=\operatorname{Tr}_B\,|\psi\rangle\langle\psi|$ &
Reduced potentiality matrix $\Psi_A=[\psi(x,y)]$,
$\rho_A=\Psi_A\Psi_A^\dagger$ \eqref{eq:reduced_potentiality_matrix} \\[3pt]
Entangled pure state (Schmidt rank $\geq2$) &
Non-factorizable joint state, $\operatorname{rank}[\psi(x,y)]\geq2$
(Theorem~\ref{new:thm:invariance}) \\[3pt]
Environmental decoherence &
Record overlaps $\gamma_{ij}$ \eqref{eq:record_overlaps} \\
\midrule
\multicolumn{2}{@{}l}{\emph{Potentiality-level notions with no named standard
counterpart}}\\
\midrule
--- (amplitudes are summed over a subspace in calculations, but the sum is not an
object of the formalism) &
Coherent aggregate $\mathcal P(A)=\sum_{\omega\in A}\psi(\omega)$
\eqref{eq:complex_measure_density}: the transport interface of an event \\[3pt]
--- (only the recorded case is named, as the L\"uders rule) &
Algebraic conditional $\psi_{\rm alg}(y\mid x)=\psi(x,y)/\psi_X(x)$
\eqref{eq:algebraic_conditional_potentiality}: conditioning appropriate to an
\emph{unrecorded} intermediate \\[3pt]
--- (a resolution of the identity inserted in a propagator, with no probabilistic
name) &
Law of total potentiality \eqref{eq:law_total_potentiality} as kernel composition \\[3pt]
--- (the fact is expressible --- an entangled state with accidentally
factorizing statistics in one basis, cf.\ the Remark following
Theorem~\ref{new:thm:invariance} --- but carries no independence-theoretic
name) &
The independence gap: Born statistics may factorize in a given context while
$\operatorname{rank}[\psi(x,y)]\geq2$; context-universal factorization, or an informationally complete reconstruction, certifies product structure
(Theorem~\ref{new:thm:invariance}) \\[3pt]
--- (all decompositions of $\rho$ are identified) &
Branch structure of a mixture: the columns of $\Psi$ as potentiality branches
relative to a specified preparation, record, or remote context
(Sections~\ref{sec:mixtures} and~\ref{subsubsec:reduced_states}) \\
\botrule
\ifarxiv
\end{tabular}
\else
\end{tabular*}
\fi
\ifarxiv
\end{table*}
\else
\end{table}
\fi

\subsection{Possible physical and mathematical extensions}

Although the framework developed here is empirically equivalent to
standard quantum mechanics, it also suggests several directions for
further development. 

A first direction concerns stochastic and measurement-conditioned dynamics. Nelson's stochastic mechanics attempts to recover quantum
mechanics from an underlying diffusive dynamics
\cite{Nelson1966,Nelson1985}, but faces well-known difficulties with
repeated measurements and multi-time correlations
\cite{BlanchardGolinServa1986}. The potentiality formulation suggests a
different perspective in which stochasticity is placed at the level of the
potentiality state itself and measurement records generate conditional
updating, or filtering, under continuous observation. Developing such a
stochastic extension, and determining the conditions under which it reduces
to standard quantum trajectories, is left for future work. 

A second direction concerns time. If measurement is conditioning on
actualized records, the accumulation of records is itself a natural
candidate for an internal chronology. This direction is developed in
companion work \cite{MartelliniTime2026}.

Finally, the framework raises representation-theoretic questions. In
the standard formulation, Hilbert space is introduced as a primitive
mathematical description. In the potentiality formulation, one may instead ask
whether Hilbert space can be derived as the canonical representation
space of coherent, square-normalized, complex-valued potentiality
assignments. Gleason's theorem already suggests that probability
assignments on the quantum lattice are represented by density operators
\cite{Gleason1957}. A natural extension of the present program is to
characterize which coherence conditions on complex potentiality
assignments force the usual Hilbert-space formalism, and which
relaxations might lead to structures beyond standard quantum mechanics.

\section{Conclusion}
\label{sec:conclusion}

This paper has proposed a reformulation of quantum mechanics as a theory
of unresolved physical potentialities. From the mathematical standpoint, many of the structures usually associated
with classical probability theory are not abandoned by quantum
mechanics, but merely displaced to a deeper level. Additivity,
conditioning, independence, mixtures, transition kernels, and temporal
divisibility all have natural analogues at the level of potentialities. 

Within this framework, the standard elements of quantum mechanics
receive a unified interpretation. Pure states are coherent
potentiality distributions. Mixed states are statistical mixtures of
potentiality distributions, or equivalently potentiality matrices whose
observable content is a density matrix. The density matrix can in turn
be read as a coherence kernel: its off-diagonal blocks encode the retained
capacity of relative phase to affect future statistics, while their
vanishing yields phase invariance and incoherent composition in the
reduced description. Measurement is conditioning of potentiality on an
actualized record. Decoherence is the suppression of coherent
recombination through environmental records. Entanglement is the non-factorization of joint potentialities.

Our central conclusion is therefore that quantum mechanics can be read as a
theory in which the linear and additive structure of uncertainty
survives, not at the level of probabilities over actual outcomes, but at
the level of unresolved complex potentialities. It thereby identifies more explicitly where the characteristic
nonclassical features of quantum theory enter, while preserving the
standard formalism and all of its empirical predictions. The resulting formulation is empirically equivalent to standard quantum
mechanics. Its value is to make explicit a mathematical and conceptual layer
that is not explicit in the usual Hilbert-space representation, which is reinterpreted as the canonical
geometry of square-normalized complex potentialities. 

An open
program is to clarify whether coherent square-normalized complex
potentiality assignments, together with suitable consistency conditions
across contexts, force the Hilbert-space formalism as their canonical
representation. If this program could be carried out, Hilbert space would no longer appear as a primitive postulate, but as the
representation space naturally associated with a theory of complex
potentialities. 

\section*{Statements and Declarations}

\ifarxiv
  \paragraph{Funding.} The author received no specific funding for this work.

  \paragraph{Competing interests.} The author declares no competing interests.

  \paragraph{Data availability.} Data sharing is not applicable to this article as no datasets were generated or analyzed during the current study.

  \bibliographystyle{apsrev4-2}
  \bibliography{referencesP1}
\else
  \bmhead{Funding}
  The author received no specific funding for this work.

  \bmhead{Competing interests}
  The author declares no competing interests.

  \bmhead{Data availability}
  Data sharing is not applicable to this article as no datasets were generated or analyzed during the current study.

  \bibliography{referencesP1}
\fi

\end{document}